\documentclass[runningheads]{llncs}
\usepackage{amsmath,amssymb}
\usepackage{graphicx}
\usepackage{latexsym,graphicx}
\usepackage{xcolor}
\usepackage{url}
\usepackage{algorithm}
\usepackage{algorithmic}
\usepackage{verbatim}

\usepackage{url}
\usepackage{mathrsfs}

\usepackage{tikz}
\usepackage{calc}
\usetikzlibrary{arrows,automata,positioning,shapes}
\tikzstyle{rstate}=[state,ellipse]
\tikzset{>={latex}}

\spnewtheorem{observation}[definition]{Observation}{\bfseries}{\itshape}

\newcommand{\bigo}{{\mathcal{O}}}

\newcommand{\eps}{\epsilon}

\begin{document}

\title{Testing $k$-binomial equivalence\thanks{The results presented in this paper were partly obtained during the Dagstuhl seminar 14111, in March 2014. Dominik D. Freydenberger was supported by the DFG grant FR 3551/1-1. Pawe{\l} Gawrychowski is currently holding a post-doctoral position at Warsaw Center of Mathematics and Computer Science. Juhani Karhum\"aki was supported by Academy of Finland under the grant 257857. Florin Manea was supported by the DFG grant 596676. Wojciech Rytter was supported by the grant NCN2014/13/B/ST6/00770 of the Polish Science Center. }
}
\titlerunning{Testing $k$-binomial equivalence}

\author{Dominik D. Freydenberger\inst{1} \and Pawe{\l} Gawrychowski\inst{2,3} \and Juhani Karhum\"aki\inst{4} \and \mbox{Florin Manea\inst{5}} \and Wojciech Rytter\inst{2}}
\authorrunning{Freydenberger et al.}

\institute{
Institute for Computer Science, University of Bayreuth, Germany\\ \email{ddfy@ddfy.de}\\
\and Institute of Informatics, University of Warsaw, Poland\\ \email{$\{$gawry,rytter$\}$@mimuw.edu.pl}\\
\and Institute of Computer Science, University of Wroc{\l}aw, Poland\\
\and Department of Mathematics and Statistics, University of Turku, Finland\\ \email{karhumak@utu.fi}\\
\and Department of Computer Science, Kiel University, Germany\\ \email{flm@informatik.uni-kiel.de}
}

\maketitle

\begin{abstract}
Two words $w_1$ and $w_2$ are said to be $k$-binomial equivalent if every non-empty word $x$ of length at most $k$ over the alphabet of $w_1$ and $w_2$ appears as a scattered factor of $w_1$ exactly as many times as it appears as a scattered factor of $w_2$. We give two different polynomial-time algorithms testing the $k$-binomial equivalence of two words. The first one is deterministic (but the degree of the
corresponding polynomial is too high) while the second one is randomised (but more direct and efficient). 
\end{abstract}

\section{Introduction}
An {\it alphabet} is a finite and nonempty set of symbols (also called letters). Any finite sequence of symbols from an alphabet $\Sigma$ is called a {\it word} over $\Sigma$. The set of all words over $\Sigma$ is denoted by $\Sigma^*$ and the {\em empty word} is denoted by $\eps$; also $\Sigma^+$ is the set of non-empty words over $\Sigma$, $\Sigma^k$ is the set of all words over $\Sigma$ of length exactly $k$, while $\Sigma^{\le k}$ is the set of all words over $\Sigma$ of length at most $k$. Given a word $w$ over an alphabet $\Sigma $, we denote by $|w|$ its length; for some $1 \leq i \leq |w|$ we denote the $i$-th letter of $w$ by $w[i]$. We also denote the factor that starts with the $i$-th letter and ends with the $j$-th letter in $w$ by $w[i..j]$. For $w,x\in \Sigma^+$ we denote by $|w|_x$ the number of distinct occurrences of $x$ as a factor of $w$. 

A \emph{scattered factor} of $w\in \Sigma^*$ is a word $w[i_1]\cdots w[i_k]$ for some $k\geq 1$ such that $i_j<i_{j+1}$ for all $1\leq j\leq k-1$. The \emph{binomial coefficient} of $u$ and $v$, denoted $u \choose v$, equals the number of occurrences of $v$ as a scattered factor of $u$. Clearly, for $a\in \Sigma$ we have ${u \choose a}=|u|_a$, while for $x\in \Sigma^+$ with $|x|\geq 2$ it is not necessary that $|u|_x={u\choose x}$. 
For example, if $u=bbaa$ and $v=ba$ we have ${u \choose v} ={bbaa\choose ba}=4$, as $u[1]u[3]=u[2]u[3]=u[1]u[4]=u[2]u[4]=ba$; clearly, $|u|_{ba}=1$.

For more details regarding these binomial coefficients see Chapter 6, by Sakarovitch and Simon, from the handbook \cite{Loth97}.

A well known equivalence relation between words is that of abelian equivalence. Two words $w_1,w_2\in \Sigma^*$ are said to be \emph{abelian equivalent} if for all $a\in \Sigma$ we have  $|w_1|_a=|w_2|_a$; equivalently, $w_1$ and $w_2$ are abelian equivalent if they have the same Parikh vector, thus being permutations of each other. This relation was extended in~\cite{KSZ13} (see also \cite{HuKaSaSa11}), where the \emph{$k$-abelian equivalence} relation was defined. Two words $w_1,w_2\in \Sigma^*$ are said to be \emph{$k$-abelian equivalent} if for all $x\in \Sigma^{\leq k}$ we have $|w_1|_x=|w_2|_x$. Obviously, the $1$-abelian equivalence relation is the same as the abelian equivalence.

As $|w_1|_a={w_1\choose a}$, another way to generalise the abelian equivalence relation is to define the \emph{$k$-binomial equivalence} (see the conference paper \cite{RigoWORDS}, as well as its journal version \cite{rigo1}). Two words $w_1,w_2\in \Sigma^*$ are said to be \emph{$k$-binomial equivalent} if for all $x\in \Sigma^{\leq k}$ we have ${w_1 \choose x}={w_2\choose x}$; if $w_1$ and $w_2$ are $k$-binomial equivalent, we write $w_1\equiv_k w_2$. Again, it is easy to see that the $1$-binomial equivalence is the same as the abelian equivalence. Combinatorial properties of the $k$-binomial equivalence relation are studied in \cite{RigoWORDS,rigo1,rigo2}. 

Recently, in \cite{EhMaMeNo14,EMMN15} a series of algorithmic results regarding the $k$-abelian equivalence were shown. As a basic result, it was shown that one can test whether two words are $k$-abelian equivalent in linear time. Therefore, it seems natural to us to study a similar problem in the context of $k$-binomial equivalence. That is, we are interested in the following problem.

\begin{problem}\label{k-binom-equiv}
Given $w_1,w_2\in \Sigma^*$, with $|w_1|=|w_2|=n$, and $k\leq n$, decide whether $w_1\equiv_k w_2$.
\end{problem}

Our main result shows that Problem \ref{k-binom-equiv} can be solved in polynomial time. The proof of this result uses a series of known results from the theory of finite automata, which does not exploit in any way the properties of $k$-binomial equivalence. Moreover, the degree of the polynomial characterising the time complexity of this algorithm is rather high, so we do not give it explicitly. Instead, we also show a simpler and much more direct Monte-Carlo algorithm solving the same problem. Our solutions assume a basic understanding of formal languages and automata theory; for more details, see \cite{roz:han} and \cite{Sak09}.

The main motivation of studying the algorithmic properties of the $k$-binomial equivalence relation is of fundamental nature: we have a new relation on words and we are, naturally, interested in how we can effectively test whether two words are equivalent with respect to this relation. Our results are also motivated by the work done in avoidability of $k$-binomial repetitions (e.g., squares and cubes in \cite{rigo2}). Constructing infinite words avoiding consecutive occurrences of factors from the same equivalence class with respect to the $k$-binomial equivalence often requires extensive computer simulations, whose basic operation is testing whether two consecutive factors are equivalent. As the words one constructs in such simulations are getting longer and longer, so do their factors whose equivalence one needs to test; consequently, efficient algorithms for testing the equivalence of words are required.

Before moving to the main sections of this paper, we just point out that the complexity results we show here hold in the unit-cost RAM with logarithmic size memory word. In this model (which is generally used in the analysis of algorithms) we assume that, if the size of the input is $n$ (e.g., we are given a word of length $n$), each memory cell can store $\Omega(\log n)$ bits, or, in other words, that {\em the machine word size} is $\Omega(\log n)$. The instructions are executed one after another, with no concurrent operations. The model contains common instructions: arithmetic (add, subtract, multiply, divide, remainder, shifts and bitwise operations, equality testing, etc.), data movement (indirect addressing, load the content of a memory cell, store a number in a memory cell, copy the content of a memory cell to another), and control (conditional and unconditional branch, subroutine call and return). Each such instruction takes a constant amount of time.  This model allows measuring the number of instructions executed in an algorithm, making abstraction of the time spent to execute each of the basic instructions.

\section{A polynomial deterministic algorithm}

The first step we take towards solving Problem \ref{k-binom-equiv} is to construct, for a word $w$, a non-deterministic finite automaton $A_w$ that accepts exactly the scattered factors of length at most $k$ of $w$ and, moreover, has exactly $w\choose x$ paths labelled with the scattered factor $x$ of $w$. 

Let us assume that $|w|=n$; then $A_w$ has $nk+2$ states; these states are 
$$Q_w= \{(0,0)\}\cup \{(i,j)\mid 1\leq i \leq n, 1\leq j\leq k\}\cup\{(n+1,k+1)\}.$$ 
The initial state of the automaton is $(0,0)$, while every state $(i,j)$ with $0<j\leq k$ and $i\geq j$ is final. The state $(n+1,k+1)$ is an error state; this state and the initial state are the only states that are not final.

We define the transition function $\delta_w$ for all $(i,j)\in Q_w$ and all $a\in\Sigma$ by
\begin{equation*}
	\delta_w\left(\left(i,j\right),a\right)= \begin{cases}
		\left\{(\ell,j+1)\in Q_w \mid \ell>i, w[\ell]=a\right\} & \text{if this set is non-empty},\\
		\{(n+1, k+1)\} & \text{otherwise.}
	\end{cases}
\end{equation*}
See Figure \ref{fig1} for an illustration. An immediate consequence of this definition is that  $\delta_w((n+1,k+1),a)=\{(n+1, k+1)\}$ holds for all $a\in\Sigma$.
\begin{figure}
\begin{center}
  \begin{tikzpicture}[]
    \node[rstate] (top1) {$(1, j-1)$};
    \node[rstate, right=0.3 of top1] (top2) {$(2, j-1)$};
    \node[rstate, right=of top2] (top3) {$(i',j-1)$};
    \node[rstate, right=of top3] (top4) {$(i-1,j-1)$};

    \node[rstate, below right=of top2] (boss) {$(i,j)$};

    \node[rstate, below=2.5 of top1] (bot1) {$(i+1, j+1)$};
    \node[rstate, right=0.3 of bot1] (bot2) {$(i+2, j+1)$};
    \node[rstate, right=of bot2] (bot3) {$(\ell,j+1)$};
    \node[rstate, right=of bot3] (bot4) {$(n,j+1)$};
    \path (top2) -- node {$\cdots$} (top3);
    \path (top3) -- node {$\cdots$} (top4);
    \path (bot2) -- node {$\cdots$} (bot3);
    \path (bot3) -- node {$\cdots$} (bot4);
	\path[->]
		(top1) edge node[left=1.5em] {$w[i]$} (boss)
		(top2) edge node[right] {$w[i]$} (boss)
		(top3) edge node[right] {$w[i]$} (boss)
		(top4) edge node[right=1em] {$w[i]$} (boss)
		(boss) edge node[left=1em] {$w[i+1]$} (bot1)
		(boss) edge node[right] {$w[i+2]$} (bot2)
		(boss) edge node[right] {$w[\ell]$} (bot3)
		(boss) edge node[right=1em] {$w[n]$} (bot4);
  \end{tikzpicture}
\end{center}
\caption{
The definition of the transition function: all the transitions leaving $(i,j)$ and leading to a non-error state, as well as all transitions going to $(i,j)$. We have $(\ell, j+1)\in \delta_w( (i,j), w[\ell])$, with $i<\ell\leq n$ and $ j<k$.
}
\label{fig1}
\end{figure}
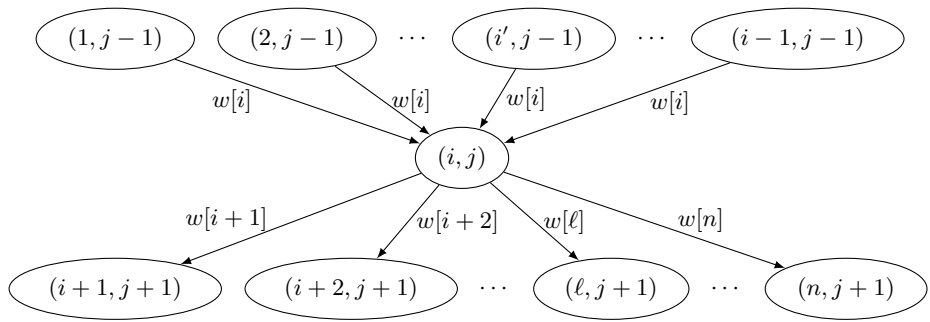

It is not hard to see that $A_w$ accepts exactly the words $w[i_1]\cdots w[i_{k'}]$ with $k'\leq k$ and $i_1< \ldots < i_{k'}$. Indeed, to accept such a word the automaton starts in the state $(0,0)$, and then goes through the states 
$$(i_1,1), (i_2,2),\ldots, (i_j,j), \ldots, (i_{k'},k');$$ 
as $1\leq i_1< \ldots < i_{k'}$ it is clear that $i_{k'}\geq k'$, so the state reached by the automaton is an accepting one. For the reverse implication, assume that the word $x$ is accepted by $A_w$ on the path formed by the states 
$$(0,0), (i_1,1), (i_2,2),\ldots, (i_j,j), \ldots, (i_{k'},k').$$ 

By the definition of $A_w$ we immediately get that $i_{j}<i_{j+1}$ for all $1\leq j\leq k'-1$; also, $i_1>0$. Thus, $i_j>j$ $1\leq j\leq k'$. Moreover, each transition ending in $(i_j,j)$ is labelled with $w[i_j]$, so $x=w[i_1]\cdots w[i_{k'}]$ is a scattered factor of $w$. 

Finally, the argument above shows that there is a bijective correspondence between the sequences of indices defining the scattered factors of length at most $k$ of $w$ and the paths of $A_w$. In conclusion, $A_w$ accepts the set of scattered factors of length at most $k$ of $w$ and, moreover, has exactly as many paths labelled with the scattered factor $x$ of $w$ as the total number of occurrences of $x$ as a scattered factor of $w$ (i.e., $w\choose x$). 

Before coming back to the solution Problem \ref{k-binom-equiv}, we recall that two non-deterministic finite automata are said to be \emph{path-equivalent} if for each word $x$ the number of distinct accepting paths labelled with $x$ of $A_1$ equals the number of distinct accepting paths labelled with $x$ of $A_2$, or both are infinite. 

In our problem, we were given $w_1$ and $w_2 $ and wanted to test whether $w_1\equiv_k w_2$. By the above, it is enough to construct $A_{w_1}$ and $A_{w_2}$ and test whether $A_{w_1}$ and $A_{w_2}$ are path-equivalent. The latter property is decidable (see \cite{siamNFA,Sak09} and the references within for a discussion on this problem and its complexity).

In the following, we show that this algorithm runs in polynomial time in our model of computation. The construction of the two automata $A_{w_1}$ and $A_{w_2}$ takes $\bigo(nk)$ time. Moreover, as none of $A_{w_1}$ and $A_{w_2}$ has transitions labelled with $\eps$, it follows from \cite{siamNFA} that there is an algorithm deciding the path-equivalence of $A_{w_1}$ and $A_{w_2}$ that runs in polynomial time with respect to the size of these automata (so, essentially, with respect to $nk$). Note that the algorithm presented in \cite{siamNFA} is only shown to run in polynomial time in a computational model where it is assumed that the arithmetic operations between any (no matter how big) rational numbers can be done in constant time. To show that the algorithm still runs in polynomial time in our model of computation, we need to go further into details. 

Basically, the algorithm of \cite{siamNFA} applied to the two automata we constructed either decides that $A_{w_1}$ and $A_{w_2}$ identifies the lexicographically first word $x$ such that $A_{w_1}$ has a different number of accepting paths labelled with $x$ than $A_{w_2}$. To do this, the algorithm explores the set of words from $\Sigma^*$ in lexicographical order; it maintains a list of words $V$ and for each $v\in V$ the array $P(v)$ storing the number of accepting paths in $A_{w_1}$ and $A_{w_2}$ (that is, an array storing for each final state of the two automata, how many paths labelled with $v$ connect the initial state of the respective automata to that final state). If the list $V$ contains at some moment the words $x_1,\ldots,x_\ell$ and the new considered word is $x$, the algorithm checks if the array $P(x)$ is linearly independent from $P(x_1), \ldots, P(x_\ell)$. If yes, $x$ is added to $V$ and the algorithm further tries all words $xa$ with $a\in \Sigma$. If no, the algor
 ithm stops trying any other word that has $x$ as a prefix. In \cite{siamNFA} it is shown that only a polynomial number of words should be tried in this process, since $V$ may contain up to $2nk$ words (as many words as the number of final states of the two automata). In our particular case, it is clear that all words that are longer than $k$ are not accepted by any of our automata (i.e., the array $P(x)$ of some $x$ longer than $k$ contains only $0$s); so, essentially, our algorithm will only try words of length at most $k+1$. Each such word $x$ that is accepted by one of our automata is accepted on at most $n^\ell$ paths, where $\ell\leq k$ is the length of $x$, in total. So, its array $P(x)$ can be stored in at most $\bigo(nk^2)$ memory words (that is, $k$ memory words for each final state, or, in other words, $k$ memory words for each component of the array). At each step of the algorithm, we test whether the newly considered $x$ produces an array $P(x)$ linearly independent from
  the arrays $P(y)$ with $y\in V$; since all these arrays contain only words that can be stored on $k$ memory words, this test can be done in polynomial time. Indeed, if we use either  a Gaussian elimination method or a modular method, such a test can be implemented in polynomial time (see, e.g., \cite{mulders} and the references within, as well as \cite{Gathen}). Finally, the algorithm just checks whether there exists a word $x$ in $V$ which is accepted on a different number of paths in $w_1$ than in $w_2$. Again, this clearly takes polynomial time. 

This concludes our analysis. We do not go into details and compute the exact complexity of the algorithm described above: we just state that it runs in polynomial time. While the preprocessing phase in which $A_{w_1}$ and $A_{w_2}$ are constructed is rather simple, computing the complexity of the algorithm from \cite{siamNFA} requires really going into the implementation details of each step (for instance, testing the linear independency of the arrays), and this is not our purpose. We just note that the exponent of $n$ in the complexity of this algorithm is at least~$3$ (in other words, the algorithm is at least cubic in $n$). The main result of our paper is, thus, the following theorem. 
\begin{theorem}Problem \ref{k-binom-equiv} can be solved in polynomial time.
\end{theorem}

Although based on a rather simple idea (the construction of the two automata), the algorithm presented in this section has a drawback: the main part of the computation is hidden in the algorithm checking the path equivalence of these two automata. Accordingly, in the following section we present a direct and more efficient randomised algorithm testing the $k$-binomial equivalence of two words.

\section{A Monte-Carlo algorithm}

We begin with a series of prerequisites. The first one is a folklore result; although it is really well known, 
we give a short sketch of the proof for completeness.
\begin{lemma}\label{prime}
We can generate a number $p$ using $\bigo(t^{3})$ operations on $t$-bit numbers, so that $p$ is a random $t$-bit prime with probability at least $1-\frac{1}{2^{t}}$.
\end{lemma}

\begin{proof}
We recall that, given a $t$-bit number $p$, one iteration of the Rabin-Miller primality test~\cite{RM} performs $\bigo(t)$ operations on $t$-bit numbers, always returns yes if $p$ is prime, and otherwise returns no with probability at least $\frac{3}{4}$. We choose a random odd $t$-bit number $p$ and execute one iteration of the Rabin-Miller test. If the test succeeds, we return $p$, and otherwise repeat. By Theorem 2 of~\cite{Pomerance}, the procedure returns a composite $p$ with probability less than $t^{2}4^{2-\sqrt{t}}$. However, we need to modify it so that the total number of operations is always $\bigo(t^{3})$.
To this end, we simply terminate after having tried $\Theta(t^{2})$ random $t$-bit numbers.
By the prime number theorem, the probability of a random odd $t$-bit number being prime is $\Theta(\frac{1}{t})$. Hence if we generate $\Theta(t^{2})$ such random numbers, the probability of all of them being composite is at most $\frac{1}{2^{t+1}}$ for $t$ large enough. Therefore, the total error probability is less than $\frac{1}{2^{t}}$ for $t$ large enough. (For smaller $t$, we can use a naive method.) The total number of operations is now always $\bigo(t^{3})$.
\qed
\end{proof}

The second auxiliary result is a particular case of the Schwartz-Zippel lemma. For a prime number $p$, let $\mathbb{F}_p$ denote the finite field on $p$ elements consisting of the integers modulo $p$. It is well known that a non-zero polynomial $Q\in \mathbb{F}_p[x]$ of degree $d$ has at most $d$ distinct roots in $\mathbb{F}_p$. Thus, the following trivially holds:
\begin{lemma}\label{SZ}
Let $Q$ be a non-zero polynomial of degree $d$ in $\mathbb{F}_p$. Then, the probability that a randomly chosen $x \in \mathbb{F}_p$ is a root of $Q$ is at most $\frac{d}{p}$:
$$\mathrm{Pr}_{x\in \mathbb{F}_p}[Q(x)=0] \leq \frac{d}{p}.$$ 
\end{lemma}

%
We now continue with the main part of this section.

For $v \in \{0,1\}^+$ let $\mathrm{bin}(v)$ be the number which binary representation is $v$. We define the crucial polynomial:
$$Q_{k,w}(x)=\sum_{|v|\le k} {w\choose v}\, x^{\mathrm{bin}(1v)}.$$
\begin{example}
\begin{eqnarray*}
Q_{2,0010}(x)&=&{0010\choose 0}\, x^{\mathrm{bin}(10)}\,+\, {0010\choose 1}\, x^{\mathrm{bin}(11)} \,+\, {0010\choose 00}\, x^{\mathrm{bin}(100)}\, + \\ 
& & {0010\choose 01}\, x^{\mathrm{bin}(101)}\,+\, {0010\choose 10}\, x^{\mathrm{bin}(110)}\,+\, {0010\choose 11}\, x^{\mathrm{bin}(111)}\\
&=& 3x^2+x^3+3x^4+x^5+x^6
\end{eqnarray*}
\end{example}
Clearly the powers of the variable $x$ encode uniquely the scattered factors of the word $w$, consequently:
\begin{observation}
$w_1\equiv_k w_2$ if and only if $Q_{k,w_1}\equiv Q_{k,w_2}$ in $\mathbb{Z}$.
\end{observation}

By definition, for any word $w$ with $|w|\geq k$, the degree of $Q_{k,w}(x)$ is $2^{k+1}-1$, so we cannot afford (time-wise) to construct it explicitly for any of the words $w_1$ and $w_2$, as enumerating the coefficients of such a polynomial would take exponential time. So, what we should see now is how to compute efficiently $Q_{w_1,w_2}(x) := Q_{k,w_{1}}(x)-Q_{k,w_{2}}(x)$ in $\mathbb{F}_{p}$; this is solved in the next lemma, where we show how $Q_{k,w}(x)$ is computed in time $\bigo(nk^{2})$ for a word $w$ of length $n$.
In the end we will choose $p$ such that $\log p = \Theta(k+\log n)$. Consequently, we assume that
operations on numbers in $\mathbb{F}_{p}$ take $\bigo(k)$ time, because in our model two numbers consisting of
$b\leq n$ bits can be added, subtracted, multiplied, and divided in $\bigo(\lceil\frac{b}{\log n}\rceil\log\lceil\frac{b}{\log n}\rceil)$ time.
The bottleneck is that we cannot construct $Q_{w_1,w_2}$ explicitly, so we need to go around this in order to compute $Q_{w_1,w_2}(x)$.
\begin{lemma}
For a word $w$ of length $n$, the value $Q_{k,w}(x)$ in $\mathbb{F}_p$ can be computed in $\mathcal{O}(k^{2}n)$ time.
\end{lemma}

\begin{proof}
We use an auxiliary polynomial. Let $$Q'_{t,w}(x)=\sum_{|v|=t}\; {w\choose v}\, x^{\mathrm{bin}(v)}.$$
In other words  $Q'_{t,w}(x)=\sum_{\ell=0}^{2^{t}-1} \alpha_{\ell,w}x^{\ell}$, where $\alpha_{\ell,w}$ is the number of
scattered occurrences of the word of length $t$ corresponding to the binary expansion of $\ell$ in the whole word $w$.
For example if $t=3$ and $w=11010$,
then $\alpha_{6,w}=4$, because $6=110_{2}$ and there are four scattered occurrences of $110$ in $11010$, i.e., ${{11010}\choose {110}}=6$.
It is enough to compute polynomials $Q'$ since
$$\sum_{t=1}^{k} x^{2^{t}} Q'_{t,w}(x) = \sum_{t=1}^k \left( \sum_{|v|=t}\; {w\choose v}\, x^{\mathrm{bin}(1v)}\right )=Q_{k,w}(x).$$
The additional factor $x^{2^{t}}$ is needed
since two different words $v,v'$ can start with different number of zeros,
so it can be the case that $\mathrm{bin}(v)=\mathrm{bin}(v')$ despite the fact that, actually, $v\neq v'$.

We use dynamic programming to compute all $Q'_{k',w[i..n]}(x)$, where $0\leq k'\leq k$, $1\leq i\leq n+1$, and $w[i..n]$ is the suffix of $w$ starting at
the $i$-th character. We denote by $T[k',i]$ the value $Q'_{k',w[i..n]}(x)$. Every such $T[k',i]$ will be computed just once and in time $\bigo(k)$ if we precompute all the numbers $x^{2^{k'}}$ for $1\leq k'\leq k$ in $\bigo(k^{2})$ time.

Then, we just have to compute $Q_{k,w}(x)=\sum_{k'=1}^{k}x^{2^{k'}}T[k',1]$, which can be done in $\bigo(k)$ time. Hence the claimed overall complexity will follow. 

First, we claim that the following recurrence holds:

$$
T[k',i]=\begin{cases}
1 & \text{if } k'=0\\
0 & \text{if } k'>0 \text{ and } i=n+1\\
T[k',i+1]+T[k'-1,i+1] & \text{if } k'>0 \text{ and } i\leq n \text{ and } w[i]=0\\
T[k',i+1]+T[k'-1,i+1]x^{2^{k'-1}} & \text{if } k'>0 \text{ and } i\leq n \text{ and } w[i]=1\\
\end{cases}
$$

This is because of the following reasoning. We write every $T[k',i]$ as a polynomial in $x$. Then, if the recurrence holds, it can be seen easily that
$T[k',i]$ is a sum, over all choices of $i\leq j_{1} < j_{2} < \dots < j_{k'}\leq n$, of monomials of the following form:
$$x^{w[j_{1}]2^{k'-1}} \times x^{w[j_{2}]2^{k'-2}} \times \dots \times x^{w[j_{k'}]2^{0}}.$$ 
But this is the same as summing monomials of the form:
$$x^{w[{j_{1}}]2^{k'-1}+w[{j_{2}}]2^{k'-2}+\dots+w[{j_{k'}}]2^{0}}.$$
Further, $w[{j_{1}}]2^{k'-1}+w[{j_{2}}]2^{k'-2}+\dots+w[{j_{k'}}]2^{0}$ is really the number from $[0,2^{k'})$ whose binary encoding is the word $w[j_{1}][j_{2}]\dots w[j_{k'}]$. 
Therefore, the coefficient of $x^{\ell}$ in $T[k',i]$ is exactly the number of ways we can choose a scattered factor $w[j_{1}][j_{2}]\dots w[j_{k'}]$ of $w[i..n]$ such that $w[j_{1}][j_{2}]\dots w[j_{k'}]$ is the binary encoding of $\ell$; in other words, this coefficient equals the number of scattered occurrences of the binary word corresponding to $\ell$ in $w[i.. n]$.

Consequently, we get that $Q'_{k',w}(x)=T[k',1]$, as claimed. The conclusion of the lemma follows easily.
\qed \end{proof}

We conclude this section by putting together all the preliminary results we have shown, to obtain the final Monte-Carlo algorithm solving Problem \ref{k-binom-equiv}.


\vskip 0.2cm  \begin{small}
    \begin{center}
    \fbox{\vspace*{0.2cm}
    \begin{minipage}{7.1cm}
    \vspace*{0.3cm} \noindent
     \hspace*{0.2cm} {\bf Randomised Algorithm}
     \vskip .2cm
\hspace*{0.2cm} let $p$ be a random $\lceil k+1+2\log n\rceil$-bit prime
\vskip .2cm
\hspace*{0.2cm}
choose random $x\in \mathbb{F}_{p}$
\vskip .2cm
\hspace*{0.2cm} compute $Q_{k,w_1}(x)$ and $Q_{k,w_2}(x)$ in $\mathbb{F}_{p}$
\vskip .2cm
\hspace*{0.2cm} $Q_{w_1,w_2}(x)\,:=\,Q_{k,w_1}(x)-Q_{k,w_2}(x)$
\vskip .2cm
 \hspace*{0.2cm} {\bf return }
 YES if $Q_{w_1,w_2}(x)\,=\,0$, NO otherwise
 \vspace*{0.2cm}
  \end{minipage}
  }
  \end{center}
  \end{small}
  
The overall time complexity is clearly polynomial both in $k$ and in $n$; as $k\leq n$, we conclude that this algorithm runs in polynomial time. More precisely, generating a prime number requires $\bigo(t^{3})$ operations, where $t=\lceil k+1+2\log n\rceil$. Then, we use $\bigo(nk)$ operations to fill the table. Therefore, the total time complexity is $\bigo(nk^{2}+(k+1+\log n)^{3}k)$. By considering the case $k\leq \log n$ and $k>\log n$ separately, we conclude that the total time complexity is $\bigo(nk^{2}+k^{4})$.

Now, if $w_1\equiv_{k} w_2$ then $Q_{w_1,w_2}(x)=0$ for all $x\in \mathbb{F}_p$, so the algorithm will always return YES.
Otherwise, there are three ways it could err. First, we could have generated a composite $p$.
This happens with probability at most $\frac{1}{2^{t}}$.
Second, it might happen that $Q_{w_{1},w_{2}}$ is non-zero in the integers, but vanishes in the integers modulo $p$.
By definition, the coefficients of $Q_{w_{1},w_{2}}$ are bounded by $n^{k}$, so if the polynomial is non-zero in the integers, yet vanishes modulo $p$, $p$ must be a prime divisor of a fixed number bounded by $n^{k}$. It is well known that the number of distinct prime divisors of $x$, denoted $\omega(x)$, satisfies $\omega(x)=\bigo(\frac{\log x}{\log\log x})$. Because there are $\pi(2^{t+1})-\pi(2^{t})=\Theta(\frac{2^{t}}{t})$ primes in the interval $[2^{t},2^{t+1})$, for $n$ large enough this happens with probability at most $\frac{\omega(n^{k})}{\pi(2^{t+1})-\pi(2^{t})}\leq k\log n \frac{k+1+2\log n}{2^{k+1}n^{2}}=o(\frac{1}{n})$.
Third, our choice of $x$ might have been unfortunate. By the Schwartz-Zippel lemma,
this happens with probability at most $\frac{2^{k+1}-1}{2^{t}}$.
By the union bound, for large enough $n$, the total error probability is, consequently, less than $\frac{1}{n}$ as required.

\begin{theorem}
Problem \ref{k-binom-equiv}, for input words of length $n$, can be solved by a Monte-Carlo algorithm with running time $\bigo(nk^2+k^{4})$. The algorithm always returns a positive answer when the input words $w_1$ and $w_2$ are $k$-binomial equivalent, and returns a negative answer when $w_1$ and $w_2$ are not $k$-binomial equivalent with probability at least $1-\frac{1}{n}$.
\end{theorem}

\section{Conclusion}
In this paper we considered the problem of deciding whether two given words $w_1$ and $w_2$ are $k$-binomial equivalent. We gave two polynomial algorithms solving this problem. The first one was deterministic, and was heavily relying on a known result showing that deciding whether two non-deterministic finite automata are path-equivalent can be done in linear time. The second one was a direct algorithm, its running time was linear in the length of the input words, but it was no longer deterministic. 

The main consequence of our result is that also finding all the factors of a long word which are $k$-binomial equivalent to a shorter one can be done in polynomial time; in other words, the problem of pattern matching under $k$-binomial equivalence can be solved in polynomial time. Indeed, one can check (using the algorithms presented in this paper) for all factors of the text whether they are $k$-binomial equivalent to the pattern and return those for which this property holds. The next theorem follows.
\begin{theorem}
Given two words $w$ and $x$ and a number $k$, we can find all the factors of $w$ that are $k$-binomial equivalent to $x$ in polynomial time.
\end{theorem}

The main open problems remaining from this work are to find simpler and more efficient algorithms solving Problem \ref{k-binom-equiv} as well as a pattern matching under $k$-binomial equivalence solution that does not use testing $k$-binomial equivalence as a subroutine.

\section*{Acknowledgements}
We thank Manfred Kufleitner and Eric Rowland for participating in the discussions on this problem during the Dagstuhl seminar 14111, that finally lead to the work presented here.

\bibliographystyle{splncs03}
\bibliography{k-binomial}

\end{document}